\documentclass [journal,onecolumn,11pt]{IEEEtran}
\usepackage{amsfonts,amsmath,amssymb}
\usepackage{indentfirst, setspace}
\usepackage{url,float}
\usepackage{color}
\usepackage[a4paper,ignoreall]{geometry}
\usepackage{longtable}
\usepackage{graphicx}
\usepackage[a4paper,ignoreall]{geometry}
\usepackage{multicol}
\usepackage{stfloats}
\usepackage{enumerate}
\usepackage{cite}
\usepackage{amsthm}
\usepackage{multirow}
\usepackage{amssymb}
 \usepackage{multirow}
\usepackage{galois}
\usepackage[all]{xy}
\usepackage[square, comma, sort&compress, numbers]{natbib}
\usepackage{url}
\usepackage{natbib,hyperref,doi}
\hypersetup{hidelinks}
\usepackage{comment}
\usepackage{verbatim}
\usepackage{tikz}
\usetikzlibrary{arrows}
\usepackage[justification=centering]{caption}
\usepackage{subfigure}
\usepackage{tikz}

 \usepackage{ltablex}\keepXColumns
\newcolumntype{Y}{>{\centering\arraybackslash}X}

\def\wang#1 {\fbox {\footnote {\ }}\ \footnotetext { From Wang: {\color{red}#1}}}
\def\tl#1 {\fbox {\footnote {\ }}\ \footnotetext { From Niu: {\color{blue}#1}}}

\newtheorem{Th}{Theorem}[section]
\newtheorem{Cor}[Th]{Corollary}
\newtheorem{Prop}[Th]{Proposition}

\newtheorem{Lem}[Th]{Lemma}
\newtheorem{Def}[Th]{Definition}
\newtheorem{example}{Example}

\newtheorem{Rem}[Th]{Remark}

\newcommand{\gf}{{\mathbb F}}

\newcommand{\fq}{{\mathbb F}_q}

\makeatletter
\newcommand{\figcaption}{\def\@captype{figure}\caption}
\newcommand{\tabcaption}{\def\@captype{table}\caption}
\makeatother

\begin{document}

\title{On the constructions of $ n $-cycle permutations}

		\author{
		{Yuting Chen, Liqi Wang and Shixin Zhu}
		\thanks{Yuting Chen, Liqi Wang and Shixin Zhu are with the School of Mathematics, Hefei University of Technology, Hefei, Anhui, China.

			The research  is partially supported by the National Natural Science Foundation of China under Grant Nos. 61802102, 61772168, the Fundamental Research Funds for the Central Universities of China under Grant Nos. JZ2020HGTB0045,	PA2019GDZC0097, and the Natural Science Foundation of Anhui Province under Grant No. 2008085QA04.

			E-mail:
     Yuting Chen: yutingc666@outlook.com,
     Liqi Wang: liqiwangg@163.com,
     Shixin Zhu: zhushixin@hfut.edu.cn
		}
	}

\maketitle{}
\begin{abstract}
Any permutation polynomial is an $ n $-cycle permutation.
When $n$ is a specific small positive integer, one can obtain efficient permutations, such as involutions, triple-cycle permutations and quadruple-cycle permutations.
These permutations have important applications in cryptography and coding theory.
Inspired by the AGW Criterion, we propose criteria for $ n $-cycle permutations, which mainly are of the form $ x^rh(x^s) $.
We then propose unified constructing methods including recursive ways and a cyclotomic way for $ n $-cycle permutations of such form.
We demonstrate our approaches by constructing three classes of explicit triple-cycle permutations with high index and two classes of $ n $-cycle permutations with low index.
\end{abstract}

\begin{IEEEkeywords}
	Finite field. permutation polynomial.  the AGW Criterion. $ n $-cycle permutation.
\end{IEEEkeywords}

\section{Introduction}

Let $\gf_q$ be the finite field with $q$ elements, where $ q $ is a prime power.
When the map $f: a \mapsto f(a) $ is a bijection on $ \gf_q $, $f(x) \in\gf_q[x]$ is called a \textit{permutation polynomial} (PP) and $f^{-1}$ denotes the compositional inverse of $ f $.
If there exists an positive integer $ n $ such that $ f^{(n)}=I $ is the identity map, we call $ f $ an \textit{n-cycle permutation}, where the $n$-th functional power of $ f $ is defined inductively by $ f^{(n)} = f \circ f^{(n-1)} = f^{(n-1)} \circ f $ and $ f^{(1)}=f, f^{(0)}=I ,f^{(-n)}=(f^{-1})^{(n)} $ with our notation.
For a small $ n $, $ n $-cycle permutations are called \textit{low-cycle permutations} in this paper.
When $n=2,3$ or $4$, $f$ is called an \textit{involution}, a \textit{triple-cycle permutation} or a \textit{quadruple-cycle permutation} respectively, and they are low-cycle permutations.

The applications of permutation polynomials in cryptography, coding theory, and combinational designs have been extensively studied, readers are refereed to \cite{MullenWang14, hou2015permutation, li2018survey,Wang2019index} and the references therein.
It is a challenging task to find new classes of permutation polynomials.
However in 2011, Akbary et al. \cite{akbary2011constructing} provided a powerful method for constructing PPs over finite fields, which was called the AGW Criterion.
Its significance lies in both providing a unified explanation of earlier constructions of PPs and serving a method to construct many new classes of PPs.
After then, several authors, such as \cite{ding2015permutation,li2017new,gupta2016some,li2017several,zha2017further,cepak2017permutations,li2018newp,wu2017permutation,xu2018some,tu2018two,li2017two}, constructed permutation polynomials of the form $x^{r}h(x^{s}) $ over $\gf_{q}$.
In many situations, both the permutation polynomial and its compositional inverse are necessary.
For instance, in block ciphers, a permutation  is used as an S-box to build the confusion layer during the encryption process.
The compositional inverse of the S-box is required while decrypting the cipher.
Recently, the compositional inverse was also applied to the study of constructing permutations with boomerang uniformity 4, please see\cite{li2019newBCT}.
Due to the importance of the compositional inverse, it attracts a lot of attentions and there are also many researches on it, such as \cite{niu2020general,wu2014compositional,wu2013compositional,tuxanidy2014inverses,li2019compositional,niu2019new,zheng2020inverses}.
Therefore, if both the permutation and its compositional inverse are efficient in terms of implementation, it is advantageous to the designer.

This motivates the use of low-cycle permutations in the S-box of block ciphers.
One immediate practical advantage of a low-cycle permutation is that the implementation of the inverse does not require much resources, which is particularly useful in devices with limited resources as a part of a block cipher.
For instance, involutions have been used frequently in block cipher designs, in AES \cite{daemen2013design}, Khazad \cite{barreto2000khazad}, Anubis \cite{barreto2000anubis} and PRINCE \cite{borghoff2012prince}.
Furthermore, low-cycle permutations (such as involutions) have been also used to construct Bent functions over finite fields \cite{mesnager2016constructions,coulter2018bent,gallager1962low} and to design codes.
Recently, in \cite{canteaut2015behaviors}, behaviours of permutations of an affine equivalent class have been analyzed with respect to some cryptanalytic attacks, and it is shown that low-cycle permutations (such as involutions)  are the best candidates against these attacks.
In addition, the study of $ n $-cycle permutations will be very helpful in classifying permutations in the view of cycle, since each permutation over finite sets must be an $ n $-cycle permutation.
Due to the importance of $ n $-cycle permutations, there are systematic studies about them in recent years.
The explicit study of involutions was started with the paper \cite{charpin2016involutions} for finite fields with even characteristic, where  basic tools and constructions of  involutions were given.
Since then, lots of attentions had been drawn in this direction.
Recently, Zheng et al. \cite{zheng2019constructions} gave a more concise criterion for involutory permutations of the form $x^rh(x^s)$ over $\gf_q$, where $s\mid {(q-1)}$.
By using this criterion, they proposed a general method, from a cyclotomic perspective, to construct involutions of such form  from given involutions over some subgroups of $\mathbb{F}_{q}^{*}$ by solving congruent and linear equations over finite fields.
Independently, Niu et al. \cite{niu2019new} started from the AGW Criterion, and proposed an involutory version of the AGW Criterion.
Then they demonstrated their results by constructing explicit involutions of the forms $ x^rh(x^s) $ and $ g\left(x^{q^i} - x+\delta\right) +cx  $.	
In 2019, \cite{liuTripleCyclePermutationsFinite2019} studied triple-cycle permutations over binary fields of the types of Monomial, Dickson polynomial and Linearized polynomial, respectively. Very recently, Wu et al. \cite{wuCharacterizationsConstructionsTriplecycle2020a} generalized the work of \cite{zheng2019constructions} and obtained some characterizations of triple-cycle permutations of the form $ x^r h(x^s ) $.
However, there are no studies on general $ n $-cycle permutations in the literature as far as we know.
This motivates us to propose unified results and to provide new explicit constructions.


The main purpose of this paper is to study the  constructions of $ n $-cycle permutations of the form $ x^rh(x^s) $ over finite sets.
First, motivated by the AGW Criterion, we generalize the previously known results about involutions in \cite{niu2019new}, and obtain that  $ n $-cycle permutations on a finite set $A$ can be constructed from $ n $-cycle permutations on its proper small set $S$ under suitable conditions.
Next, we obtain a unified criterion for $ n $-cycle permutations of the form $ x^rh(x^s) $, a method for constructing $ n $-cycle permutations from ones over the subfield recursively, and a cyclotomic perspective to characterize the properties of $ n $-cycle permutations.
Finally, we consider a piecewise method in the construction of $ n $-cycle permutations of cyclotomic form.
We will explain our approaches by giving five classes of explicit constructions of PPs.

The rest of this paper is organized as follows.
In Section \ref{pre},  we introduce some basic knowledge about general $ n $-cycle permutations.
  Main tools for constructing $ n $-cycle permutations of the form $ x^rh(x^s) $ are proposed in Section \ref{gen}.
Explicit $ n $-cycle permutations are respectively constructed from two different perspectives in Sections \ref{high} and \ref{low}. A conclusion of this paper is given in Section \ref{cons}.

\section{Preliminaries}
\label{pre}


Before we handling permutations of the form $ x^rh(x^s) $ , we prepare and discuss general $ n $-cycle permutations in this section.
\begin{Def}
	\cite[Definition 4]{charpin2016involutions}
	Let $f$ be a permutation of a finite set $ A $.
	Let $t$ be a positive integer.
	A \textit{cycle} of $f$ is a subset $\left\{x_{1}, \cdots, x_{t}\right\}$ of pairwise distinct elements of $A$ such that $f\left(x_{i}\right)=x_{i+1}$ for $1 \leq i \leq t-1$ and $f\left(x_{t}\right)=x_{1} $.
	The cardinality of a cycle is called its \textit{length}.
\end{Def}

The result below is a generalization of 	\cite[Proposition 6]{charpin2016involutions} for $ n $-cycle permutations.
\begin{Prop}
	\label{nandc}
	Let $f$ be a permutation of a finite set $ A $.
	Then, $ f $ is an $ n $-cycle permutation if and only if the length $ l $ of each cycle of $ f $ is no more than $ n $ and $ l  \mid n $.
\end{Prop}
\begin{proof}
	Let $x_{1},x_{2},\cdots,x_{l}$ be $l$ elements defined by $x_{l}=f(x_{l-1}),x_{l-1}=f(x_{l-2}),\cdots,x_{3}=f(x_{2}),x_{2}=f(x_{1})$, and define $ x_{i}=x_{i+l} $ for any integer $ i $.
	Then $x_{l}=f(x_{l-1})=f(f(x_{l-2}))=f^{(n)}(x_{l-n})=x_{l-n},$ since $f$ is an $ n $-cycle permutation.
	One can obtain $  l-n   \equiv 0   \pmod l $.
	Thus, the length $ l $ of each cycle of an $ n $-cycle permutation is no more than $ n $, and $ l   \mid n $.
	
	Conversely, if the length $ l $ of each cycle of $ f $ is no more than $ n $ and $ l  \mid n $, then for $ x_i $ in each cycle, we have $ f^{(n)}(x_i)=x_i $, i.e., $ f $ is an $ n $-cycle permutation.
\end{proof}
Apparently if $ f $ is an $ n $-cycle permutation, then $ f $ is also an $ kn $-cycle permutation, where $ k $ is a positive integer.
Although we focus on the $ n $-cycle permutation in this paper, we propose the following results for clarity.
\begin{Def}
	For the least positive integer such that $ f^{(n)}=I $, we call $ f $ a \textit{fundamentally n-cycle permutation}.
\end{Def}
\begin{Prop}
	\label{lcm}
	The least common multiple of lengths of all cycles of $ f $ is $ n $ if and only if $ f $ is a fundamentally $ n $-cycle permutation.
\end{Prop}
\begin{proof}
According to Proposition \ref{nandc}, $ f $ is an $ n $-cycle permutation if and only if the length $ l $ of each cycle of $ f $ satisfying $ l  \mid n $.
Thus, the least common multiple of lengths of all cycles of $ f $ is exactly the least positive integer such that $ f^{(n)}=I $.
\end{proof}

\begin{Prop}
	Assume $ n $ is a prime, then $ f $ is a fundamentally $ n $-cycle permutation if and only if $ f $ is an $ n $-cycle permutation and $ f \ne I $.
\end{Prop}
\begin{proof}
According to Proposition \ref{nandc}, the length of each cycle of an $ n $-cycle permutation $ f \ne I $ is either $ 1 $ or $ n $, since $ n $ is prime.
Thus, $ f $ is a fundamentally $ n $-cycle permutation if and only if $ f $ is an $ n $-cycle permutation and $ f \ne I $, by Proposition \ref{lcm}.
\end{proof}
\begin{Rem}
	Thus, for $ 2 $-cycle permutations (involutions) and $ 3 $-cycle permutations (triple-cycle permutations), they are respectively fundamentally $ 2 $-cycle permutations and fundamentally $ 3 $-cycle permutations clearly.
\end{Rem}

Now we recall a famous result named AGW Criterion, which was proposed by  Akbary et al. for constructing PPs in \cite{akbary2011constructing}.
It will play an important role in our following results.
\begin{Lem}
	\label{LGWlemma}
	(\cite{akbary2011constructing}, AGW Criterion)
	Let $A, S$, and $\overline{S}$ be finite sets with $\# S=\# \overline{S}$, and let $f: A\to A,$ $g: S\to \overline{S}$, $\lambda: A\to S$ and $\overline{\lambda}: A\to\overline{S}$ be maps such that $\bar{\lambda}\circ f=g\circ \lambda$.
	If both $\lambda$ and $\bar{\lambda}$ are surjective, then the following statements are equivalent:
	\begin{enumerate}[(1)]
		\item $f$ is a bijection and
		\item $g$ is a bijection from $S$ to $\overline{S}$ and $f$ is injective on $\lambda^{-1}(s)$ for each $s\in S$.
	\end{enumerate}
\end{Lem}
The AGW Criterion can be illustrated as follows:
\begin{equation*}
	\xymatrix{
		A \ar[rr]^{f}\ar[d]_{\lambda} &   &  A  \ar[d]^{\overline{\lambda}} \\
		S	 \ar[rr]^{g} &  & \overline{S} .}
\end{equation*}
It not only provided a unified explanation for a lot of earlier constructions, but also motivated more new findings of PPs.
Inspired by the AGW Criterion, we obtain the following result for all $ n $-cycle permutations on a finite set, which is a generalization of \cite[Proposition 2.2]{niu2019new}.
\begin{Th}
	\label{criterion}
	Let $A$ and $S$ be finite sets, and let $f:A \to A$, $g: S \to S$, $\lambda:A \to S$ be maps such that $\lambda$ is surjective and $\lambda \circ f = g \circ \lambda$.
	Assume that $f$ is an $ n $-cycle permutation on $A$,  then $g$ is an $ n $-cycle permutation on $S$.
\end{Th}
\begin{proof}
	We have $\lambda = \lambda \circ f ^{(n)} = g \circ \lambda \circ f ^{(n-1)}=  g^{(n)} \circ \lambda$, where $ f^{(n)} $ denotes $ f $ composing itself for $ n-1 $ times.
	Since $\lambda$ is surjective, we obtain that $g$ is an $ n $-cycle permutation on $S$ from {$\lambda = g^{(n)} \circ \lambda$.	}
\end{proof}

Once we obtained an $ n $-cycle permutation, it is natural to obtain more $ n $-cycle permutations by composing itself by the following lemma.
\begin{Lem}
	\label{criterion2}
	Assume that $f$ is an $ n $-cycle permutation on $A$,  then $f^{(k)}$ is also an $ n $-cycle permutation, where $ n \ge 3,  1<k <n $.
\end{Lem}
Furthermore, for different mappings with $ n $-cycle permutations, we have the following results.
\begin{Prop}
	\label{criterion3}
	Assume that $f,g$ are $ n $-cycle permutations on $A$,  then  $f\circ g$ is also an $ n $-cycle permutation if $ f\circ g= g \circ f $.
\end{Prop}
\begin{proof}
	It is easy to verify by the definition of $ n $-cycle permutation.
\end{proof}
\begin{Prop}
	\label{criterion4}
	Assume that $ f,g $ permute $ A $. Furthermore, $f$ is an $ n $-cycle permutation.
	Then $g \circ f \circ g^{-1}$ is also an $ n $-cycle permutation.
\end{Prop}
\begin{proof}
   Since $ (g \circ f \circ g^{-1})^{(n)}= g \circ f^{(n)} \circ g^{-1}=g \circ I\circ g^{-1} =I $, $g \circ f \circ g^{-1}$ is also an $ n $-cycle permutation.
\end{proof}
The results above can be used to quickly generate large general $ n $-cycle permutations from known ones over the same finite sets.
However, in this paper we do not give examples for them, since they are simple to implement.


Now we introduce PPs of the form $ x^rh(x^s) $.
Throughout the rest of paper, let $q$ be a power of a prime, and $s, \ell$ be divisors of $q-1$ such that $s \ell=q-1 .$
All  nonzero elements of a finite set $ A $ is denoted by $ A^* $.
In this paper, we always assume that $h(x)\neq 0$ for any $x\in\mu_{\ell}$.
Otherwise, it is easy to obtain that $  f(x)=x^rh(x^s) $ can not permute $ \gf_{q} $.
It is well-known that every polynomial $f$ over $\mathbb{F}_{q}$ such that $f(0)=0$ has the form $x^rh(x^s) $ for some positive integers $r, s$ in the case $s | (q-1) $.

Based on \cite{niederreiter2005cyclotomic}, the concept of the index of a polynomial was first introduced in \cite{akbary2009permutation}.
Any nonconstant polynomial $f(x) \in \mathbb{F}_{q}[x]$ of degree $\leq q-1$ can be written uniquely as $f(x)=a\left(x^{r} h\left(x^{(q-1) / \ell}\right)\right)+b$ such that the degree of $h$ is less than the index $\ell$ which is defined in \cite{Wang2019index} in the following.
Namely, write
$$
f(x)=a\left(x^{d}+a_{d-i_{1}} x^{d-i_{1}}+\cdots+a_{d-i_{k}} x^{d-i_{k}}\right)+b,
$$
where $a, a_{d-i_{j}} \neq 0, i_{0}=0<i_{1}<\cdots<i_{k}<d, j=1, \ldots, k $.
The case $k=0$ is trivial and we have $\ell=1 $.
Thus we shall assume that $k \geq 1 $.
Write $d-i_{k}=r,$ the vanishing order of $x$ at 0 (i.e., the lowest degree of $x$ in $f(x)-b$ is $r$ ).
Then $f(x)=a\left(x^{r} h\left(x^{(q-1) / \ell}\right)\right)+b,$
where $h(x)=x^{e_{0}}+a_{d-i_{1}} x^{e_{1}}+\cdots+a_{d-i_{k-1}} x^{e_{k-1}}+a_{r}, s=\gcd\left(d-r, d-r-i_{1}, \ldots, d-r-i_{k-1}, q-1\right), d-r=e_{0} s, d-r-i_{j}=e_{j} s, 1 \leq j \leq k-1,$ and $\ell:=\frac{q-1}{s} $.
Hence in this case $\gcd\left(e_{0}, e_{1}, \ldots, e_{k-1}, \ell\right)=1 $.
The integer $\ell =\frac{q-1}{s}$ is called the \textit{index} of $ f(x)= x^r h(x^s) $.
One can see that the greatest common divisor condition in the definition of $s$ makes the index $\ell$ minimal among those possible choices\cite{Wang2019index}.
Note that the index of a polynomial is closely related to the concept of the least index of a cyclotomic mapping polynomial \cite{niederreiter2005cyclotomic}.
Let $\mu_{\ell}$ denote the set of $\ell$-th roots of unity in $\mathbb{F}_{q}^{*},$ i.e.
$$ \mu_{\ell}=\left\{x \in \mathbb{F}_{q}^{*} | x^{\ell}=1\right\}, $$
which is also the unique cyclic subgroup of $\mathbb{F}_{q}^{*}$ of order $\ell$.
This subgroup can also be represented as $\left\{x^{s} | x \in \mathbb{F}_{q}^{*}\right\}$.
It is easy to verify that
$$  \mathbb{F}_{q}^{*}=\bigcup_{i=1}^{\ell} S_{\alpha_{i}}, \quad S_{\alpha_{i}}=\left\{x \in \mathbb{F}_{q}^{*} | x^{s}=\alpha_{i}\right\}, $$
where $\alpha_{i} \in \mu_{\ell}$ for $1 \leq i \leq \ell $.
Recall that $f(x)=x^{r} h\left(x^{s}\right)$ is called an $r$-th order cyclotomic mapping \cite{wang2007cyclotomic,niederreiter2005cyclotomic}.
Specifically, it can be expressed in the form of piecewise map as follows:
\begin{equation}
	\label{fenyuan}
	f(x)=x^{r} h\left(x^{s}\right)=\left\{\begin{array}{lr}
		0, & x=0 \\
		h\left(\alpha_{i}\right) x^{r}, & x \in S_{\alpha_{i}}, 1 \leq i \leq \ell
	\end{array}\right.
\end{equation}
where $\alpha_{i} \in \mu_{\ell}, i=1,2, \ldots, \ell $.
For more information about the index approach and piecewise method, readers are refereed to \cite{Wang2019index,caoConstructingPermutationPolynomials2014} et al.
The following well-known lemma, the multiplicative form of the AGW Criterion, gives a necessary and sufficient condition for $x^{r} h\left(x^{s}\right)$ being a permutation over $\mathbb{F}_{q}$.
\begin{Lem}
	\label{mulagw}
	\cite[Theorem 2.3]{park2001permutation}
	\cite[Theorem 1]{wang2007cyclotomic}
	\cite[Lemma 2.1]{zieve2009some}
	\label{Cri}
	$f(x)=x^rh\left(x^s\right) \in\gf_q[x]$ permutes $\gf_q$ if and only if
	\begin{enumerate}[(1)]
		\item $\gcd\left(r,s\right)=1$ and
		\item $g(x)=x^rh(x)^s$ permutes $\mu_{\ell},$ where $\mu_{\ell}=\left\{x\in{\gf}_{q}^*  \ | \ x^{\ell}=1\right\}$.
	\end{enumerate}
\end{Lem}
Lemma \ref{mulagw}, independently obtained by several authors earlier \cite{park2001permutation,wang2007cyclotomic,zieve2009some}, was used frequently  to study PPs of the form $x^rh(x^{(q-1)/\ell})$ over $\fq$, where $\ell \mid (q-1)$.

In the rest of the paper, we focus on the constructions of $ n $-cycle permutations of the form $ x^rh(x^s) $, 
which will be helpful for considering general $ n $-cycle permutations of other forms.

\section{General Results for $ n $-Cycle Permutations over Finite Fields}
\label{gen}

In this section, we give a general approach to $ n $-cycle permutations of the form $ x^rh(x^s) $.

First, we give a concrete criterion, which is a generalization of \cite[Theorem 3]{zheng2019constructions} and \cite[Proposition 2.6]{niu2019new}.

\begin{Th}
	\label{mulcoren}
	Let $q$ be a prime power and $f(x)= x^rh\left( x^s\right)  \in \gf_q[x]$, where $s \mid (q-1),\gcd(r,s)=1$.
	Assume that $g(x)=x^rh(x)^s$ is a polynomial on $\mu_\ell=\left\{    x\in{\gf}_{q}^*  \    |  \   x^\ell=1   \right\}$, where $  \ell = {(q-1)/s} $.
	Then, $f$ is an $n$-cycle permutation over $\mathbb{F}_{q}$ if and only if
	\begin{enumerate}[(1)]
		\item $r^{n} \equiv 1 \bmod s$ and
		\item $\varphi(y)=y^{(r^n-1)/s}  \prod\nolimits_{{i = 0}}^{n-1} {   h\left(g^{(i)}(y) \right)^{r^{n-i-1}}      }          =1 $ for all $y \in \mu_{\ell}$.
	\end{enumerate}
\end{Th}
\begin{proof}
	First of all, $ f(0)=0 $. Then we only need to consider the non-zero case below.
	
	Assume that $r^{n} \equiv 1 \bmod s$ and $\varphi(y)=  y^{(r^n-1)/s}  \prod\nolimits_{{i = 0}}^{n-1} {   h\left(g^{(i)}(y) \right)^{r^{n-i-1}}      }        =1  $ for all $y \in \mu_{\ell}$.
	For any $x \in \gf_q^*$, one can find a $ y=x^s \in \mu_{\ell} $ such that $ y^{(r^n-1)/s}=x^{r^n-1} $.
	Then we have
	\begin{equation*}
		\begin{aligned}
			f^{(n)}(x)=&   	(f^{(n-1)}(x) )^rh\left(  (	f^{(n-1)}(x) )^s\right)   \\
			=&   	\left((f^{(n-2)}(x))^rh((f^{(n-2)}(x) )^s) \right)^rh\left( g (	f^{(n-2)}(x) )^s\right)   \\
			=&  (f^{(n-2)}(x))^{r^2} \prod\nolimits_{{i = 0}}^{1} {   h\left(g^{(i)}((f^{(n-2)}(x))^s)\right)^{r^{1-i}}      }    \\
			=&  \  ... \\
			=&f(x)^{r^{n-1}}  \prod\nolimits_{{i = 0}}^{n-2} {   h\left(g^{(i)}(f(x)^s)\right)^{r^{n-i-2}}      }  \\
			=& (x^{r} h(x^s))^{r^{n-1}}     \prod\nolimits_{{i = 1}}^{n-1} {   h\left(g^{(i-1)}(f(x)^s)\right)^{r^{n-i-1}}      }    \\
			=& x^{r^n} h(x^s)^{r^{n-1}}      \prod\nolimits_{{i = 1}}^{n-1} {   h\left(g^{(i)}(x^s)\right)^{r^{n-i-1}}      }    \\
			=& x^{r^n}  \prod\nolimits_{{i = 0}}^{n-1} {   h\left(g^{(i)}(x^s)\right)^{r^{n-i-1}}      }    \\
			=& x , \\
		\end{aligned}
	\end{equation*}
	which indicates that $f(x)=x^rh(x^s)$ is an $ n $-cycle permutation over $\mathbb{F}_{q}$.
	
	Assume that $ f $ is an $ n $-cycle permutation over $\mathbb{F}_{q}$.
	For each $ y \in \mu_{\ell} $ and any $ x \in \gf_q^* $ such that $ x^s=y $,  we have
	$$x^{r^n-1}  \prod\nolimits_{{i = 0}}^{n-1} {   h\left(g^{(i)}(y)\right)^{r^{n-i-1}}      }  =1  . $$
	Suppose that there is an $ x_0 \in \gf_q^* $ such that $ x_0^s=y $, then we have
	$$x_0^{r^n-1}  \prod\nolimits_{{i = 0}}^{n-1} {   h\left(g^{(i)}(y)\right)^{r^{n-i-1}}      }  =1  . $$
	Hence, $ {\left(\frac{x}{x_0}\right)}^{r^n-1} =1 $ holds for any $ x \in \{ x \in \gf_q^*  \  | \ x^s=y  \} $, which implies that $r^{n} \equiv 1 \bmod s$.
	Therefore,
	\begin{equation*}
		\begin{aligned}
			\varphi(y) =&y^{(r^n-1)/s}  \prod\nolimits_{{i = 0}}^{n-1} {   h\left(g^{(i)}(y) \right)^{r^{n-i-1}}      }      \\
		    =& 	x_0^{r^n} / x_0 \prod\nolimits_{{i = 0}}^{n-1} {   h\left(g^{(i)}(x_0^s)\right)^{r^{n-i-1}}      }    \\
			=& f^{(n)}(x_0) / x_0 \\
			=& 1 . \\
		\end{aligned}
	\end{equation*}
	The last equation holds since $ f $ is an $ n $-cycle permutation over $\mathbb{F}_{q}$.
\end{proof}

The following result is a direct consequence of the above for $ n=3 $, which was given in \cite{wuCharacterizationsConstructionsTriplecycle2020a} very recently.
\begin{Cor}
	\label{mulcore}
	\cite[Theorem 1]{wuCharacterizationsConstructionsTriplecycle2020a}
	Let $q$ be a prime power and $f(x)= x^rh\left( x^s\right)  \in \gf_q[x]$, where $s \mid (q-1),\gcd(r,s)=1$.
	Assume that $g(x)=x^rh(x)^s$ is a polynomial on $\mu_\ell=\left\{    x\in{\gf}_{q}^*  \    |  \   x^\ell=1   \right\}$, where $  \ell = {(q-1)/s} $.
	Then, $f$ is a triple-cycle permutation over $\mathbb{F}_{q}$ if and only if
	\begin{enumerate}[(1)]
		\item $r^{3} \equiv 1 \bmod s$ and
		\item $\varphi(y)=y^{(r^3-1)/s}h(y)^{r^2}h\left( g\left( y\right) \right)  ^rh\left( g\left( g(y) \right) \right)  =1 $ for all $y \in \mu_{\ell}$.
	\end{enumerate}
\end{Cor}

Theorem \ref{mulcoren} provides a useful tool for both determining and constructing $ n $-cycle permutations of the form $ x^rh(x^s) $.
In order to construct $ n $-cycle permutations more efficiently, we further look for more tools in different forms for $ n $-cycle permutations.
A consequence of Theorem \ref{criterion} is the following.
It is a necessary condition for triple-cycle permutations.
\begin{Cor}
	\label{criterionmul}
		\cite[Corollary 1]{wuCharacterizationsConstructionsTriplecycle2020a}
	Assume that $f(x)=x^rh(x^s)$ is a triple-cycle permutation over $\mathbb{F}_{q}$.
Then $g(x)=x^{r} h(x)^{s}$ is a triple-cycle permutation on $\mu_{\ell}$.
\end{Cor}
Thus, to obtain $ n $-cycle permutations of the form $f(x)=x^{r} h\left(x^{s}\right)$ over $\mathbb{F}_{q}$, it is crucial to find a suitable $h(x)$.
According to Theorem \ref{criterion}, $ g(x)=x^{r} h(x)^{s}$ is necessary to be an $ n $-cycle permutation on the subgroup $\mu_{\ell}$ of $\mathbb{F}_{q}^{*}$ for $f(x)$ being an $ n $-cycle permutation on $\mathbb{F}_{q} .$

Below we represent approaches to obtain many new $ n $-cycle permutations reductively over finite fields from the known ones on their subfields.
They are obtained by assuming that $\mu_{\ell}$ is a subfield of $\mathbb{F}_{q}$.
\begin{Th}
	\label{v=1}
	Let $q$ be a prime power.
	Assume that $h(x)\in \gf_{q^n}[x]$ such that $h(y)^{q-1}=y^{1-q}$ holds for any $y \in \mu_\ell=\left\{    x\in{\gf}_{q^n}^*  \    |  \   x^\ell=1   \right\}$ , where $  \ell = {(q^n-1)/(q-1)} $.
	Then $f(x)= x^qh(x^{q-1}) $ is an $ n $-cycle permutation over $\mathbb{F}_{q^n}$ if and only if $h\left(\mu_{\ell}\right) \subset \mu_{\ell}$.
\end{Th}
\begin{proof}
		Let $r=q,s=q-1$.
	Then we have	$g(y)=y$ by
    plugging all the conditions given into Theorem \ref{mulcoren}.
	Hence, $f(x)= x^qh(x^{q-1}) $ is an $ n $-cycle permutation over $\mathbb{F}_{q^n}$  if and only if for any $y \in \mu_\ell$,
	\begin{equation}
		\label{pd1}
		y^{(q^n-1)/{(q-1)}}  \prod\nolimits_{{i = 0}}^{n-1} {   h\left( y  \right)^{q^{n-i-1}}      }     =1 .
	\end{equation}
	 Eq. (\ref{pd1}) holds if and only if $h\left(\mu_{\ell}\right) \subset \mu_{\ell}$, due to the fact that $ \ell= \sum\nolimits_{{i = 0}}^{n-1} {  q^{i}      }  $ .
	Thus $f $ is an $ n $-cycle permutation over $\mathbb{F}_{q^n}$ if and only if $h\left(\mu_{\ell}\right) \subset \mu_{\ell}$.
\end{proof}

\begin{Th}
	\label{digui}
	Let $q$ be a prime power and $f(x)= x^rh\left( x^s\right)  \in \gf_q[x]$, where $s \mid (q-1),\gcd(r,s)=1$ , $r^{n} \equiv 1 \bmod s$, $ \gcd(s,\ell)=1 $ and $\ell=\frac{q-1}{s} $ .
	Let $h(x) \in \mathbb{F}_{q}[x]$ satisfy $h\left(\mu_{\ell}\right) \subset \mu_{\ell} $.
	Then $f(x)=x^{r} h\left(x^{s}\right)$ is an $ n $-cycle permutation on $\mathbb{F}_{q}$ if and only if $g(x)=x^{r} h(x)^{s}$ is an $ n $-cycle permutation on $\mu_{\ell}$.
\end{Th}
\begin{proof}
	Assume that $g(x)$ is an $ n $-cycle permutation on $\mu_{\ell}$.
	Then for any $ y \in \mu_{\ell} $, we have
	\begin{equation*}
		\begin{aligned}
			\varphi(y)^s&=y^{r^n-1}  \prod\nolimits_{{i = 0}}^{n-1} {   h\left(g^{(i)}(y) \right)^{sr^{n-i-1}}      }   \\
			&=y^{-1}  (  g(y)   )^{r^{n-1}}  \prod\nolimits_{{i = 1}}^{n-1} {   h\left(g^{(i)}(y) \right)^{sr^{n-i-1}}      }   \\
			&=y^{-1}  (  g^{(2)}(y)   )^{r^{n-2}}  \prod\nolimits_{{i = 2}}^{n-1} {   h\left(g^{(i)}(y) \right)^{sr^{n-i-1}}      }   \\
			&= \  ... \\
			&=y^{-1}    (  g^{(n-2)}(y)   )^{r^2}  h\left( g^{(n-2)}(y)  \right)^{sr}    h\left( g^{(n-1)}(y)  \right)^s   \\
			&=y^{-1}    (  g^{(n-1)}(y)   )^r    h\left( g^{(n-1)}(y)  \right)^s   \\
			&=y^{-1}   g^{(n)}( y  )  \\
			&=1.  \\
		\end{aligned}
	\end{equation*}
	Since $ \gcd(s,\ell)=1 $, one can obtain $\varphi(y)=1$.
	Thus $f$ is an $ n $-cycle permutation on $\mathbb{F}_{q}$ according to Theorem \ref{mulcoren}.
	Conversely, $g(x)=x^{r} h(x)^{s}$ is an $ n $-cycle permutation on $\mu_{\ell}$ if $f$ is an $ n $-cycle permutation on $ \gf_{q} $, according to Theorem \ref{criterion}.
\end{proof}

If we choose a special $s$ in Theorem \ref{digui} such that $\mu_{\ell}$ is exactly a multiplicative group of a finite field, then the following corollary is obtained.
\begin{Cor}
	\label{diguiusen}
	Let $q$ be a prime power and $m, r$ be positive integers with $\gcd(q-1, m)=1$ and $r^{n} \equiv 1 \bmod \frac{q^m-1}{q-1} $.
	Let $h(x) \in \mathbb{F}_{q}[x] .$
	Then $f(x)=x^r h\left(x^{\frac{q^{m}-1}{q-1}}\right)$ is an $ n $-cycle permutation on
	$\mathbb{F}_{q^{m}}$ if and only if $g(x)=x^{r} h(x)^{m}$ is an $ n $-cycle permutation on $\mathbb{F}_{q}$.
\end{Cor}
\begin{proof}
	Let $ s= \frac{q^{m}-1}{q-1}$ in Theorem \ref{digui}.
	Since $s=m+\sum_{i=1}^{m-1}\left(q^{i}-1\right) \equiv	m \bmod (q-1)$ and $h(x) \in \mathbb{F}_{q}[x]$,  one can obtain that $ \gcd(q-1, s)=\gcd(q-1, m)=1  $  and  $ g(x)=x^{r} h(x)^{\frac{q^{m}-1}{q-1}}=x^{r} h(x)^{m} $   for any $ x \in \gf_q $.
	According to Theorem \ref{digui}, $f$ is an $ n $-cycle permutation on
	$\mathbb{F}_{q^{m}}$ if and only if $x^{r} h(x)^{m}$ is an $ n $-cycle permutation on $\mathbb{F}_{q}$.
\end{proof}
A consequence of Corollary \ref{diguiusen} for $ n=3 $ is the following.
\begin{Cor}
	\label{diguiuse}
		\cite[Corollary 4]{wuCharacterizationsConstructionsTriplecycle2020a}
	Let $q$ be a prime power and $m, r$ be positive integers with $\gcd(q-1, m)=1$ and $r^{3} \equiv 1 \bmod \frac{q^m-1}{q-1} $.
	Let $h(x) \in \mathbb{F}_{q}[x] .$
	Then $f(x)=x^r h\left(x^{\frac{q^{m}-1}{q-1}}\right)$ is a triple-cycle permutation on
	$\mathbb{F}_{q^{m}}$ if and only if $g(x)=x^{r} h(x)^{m}$ is a triple-cycle permutation on $\mathbb{F}_{q}$.
\end{Cor}

Corollary \ref{diguiusen} allows us to construct new explicit $ n $-cycle permutations over finite fields from the known ones on their subfields.
Examples will be given in the next section (see Corollarys \ref{diguiuse1} and \ref{diguiuse2} ).

In the perspective of cyclotomic, we propose a general method to determine the coefficients of $h(x)$ from a given $ n $-cycle permutation $g(x)$ over $\mu_{\ell}$, which can be seen as a generalization of \cite[Theorem 6]{zheng2019constructions} and \cite[Theorem 3]{wuCharacterizationsConstructionsTriplecycle2020a}.
\begin{Th}
	\label{piecewisegenerel}
Let $\beta$ be a primitive element of $\mathbb{F}_{q}$ and $\omega=\beta^{s}$ be a generator of the subgroup $\mu_{\ell}$.
For any $ 0\le i \le \ell-1 $, $ 0 \le j \le n-1 $, and $k  \equiv  j \bmod n  $, let $a_{(0,j)}, a_{(1,j)}, \ldots, a_{(i,j)}, \ldots, a_{(\ell-1,j)}$ be rearrangements of $0,1, \ldots, \ell-1$, satisfying $ a_{(i,0)} =i, \ell_{a_{(i,j)}} =a_{(i,j+1)}$ and $ a_{(i,k)} =a_{(i,j)} $.
Assume that  $0 \leq$ $m_{i}, m_{a_{(i,j)}} \leq s-1$ are integers for $ 0\le i \le \ell-1 $, and $ h(x)=\sum_{i=0}^{\ell-1} h_i x^i \in \gf_{q}[x] $ is a reduced polynomial modulo $x^{\ell}-1$, such that
\begin{equation}
	\label{matrix}
	H=A^{-1}B,
\end{equation}
where
\begin{equation*}
	\label{matrixA}
	A=\left(
	\begin{array}{cccc}
		1 & 1 & \cdots & 1 \\
		1 & \omega & \cdots & \omega^{\ell-1} \\
		1 & \omega^{2} & \cdots & \omega^{2(\ell-1)} \\
		& \cdots & \cdots & \\
		1 & \omega^{\ell-1} & \cdots & \omega^{(\ell-1)(\ell-1)}
	\end{array}
	\right)
	\text{is a Vandermonde matrix,}
\end{equation*}
\begin{equation*}
	H=\left(\begin{array}{c}
		h_{0} \\
		h_{1} \\
		h_{2} \\
		\vdots \\
		h_{\ell-1}
	\end{array}\right)
	\text { and }
	B=\left(\begin{array}{c}
		\beta^{\ell m_{0}+a_{(0,1)}} \\
		\vdots \\
		\beta^{\ell m_{i}+a_{(i,1)}-i r} \\
		\vdots \\
		\beta^{\ell m_{\ell-1}+a_{(\ell-1,1)}-(\ell-1) r}
	\end{array}\right) .
\end{equation*}
	Then $  f (x) = x^r h(x^s) $ is an $n$-cycle permutation over $ \gf_q $ if and only if
\begin{equation}
	\label{solution}
	\left\{\begin{aligned}
		r^{n} \equiv & 1 \bmod s   \\
		 \sum\nolimits_{k = 0}^{n-1} {r^{n-k-1}  m_{a_{(i,k)}}  }   \equiv & 0 \bmod s  . \\
	\end{aligned}\right.
\end{equation}
\end{Th}
\begin{proof}
According to $ AH=B $ and  $ \ell_{a_{(i,j)}} =\ell_{a_{(i,j+1)}},$ for $ 0\le i \le \ell-1 $, $0 \le j \le n-1 $. One can obtain that
\begin{equation*}
\left\{\begin{aligned}
	h\left(\omega^{i}\right) &=\beta^{\ell m_{i}+a_{(i,1)}-i r} \\
	h\left(\omega^{a_{(i,j)}}\right) &=\beta^{\ell m_{a_{(i,j)}}+a_{(i,j+1)}-a_{(i,j)} r}  \\
	h\left(\omega^{a_{(i,n-1)}}\right) &=\beta^{\ell m_{a_{(i,n-1)}}+i-a_{(i,n-1)} r}  .\\
\end{aligned}\right.
\end{equation*}
Thus $g(x)=x^{r} h(x)^{s}$ is an $n$-cycle permutation on $\mu_{\ell}$, i.e., for $0 \leq i \leq \ell-1, 0 \le j \le n-1$, we have
\begin{equation*}
	\label{geq}
	\left\{
	\begin{aligned}
		g\left(\omega^{i}\right) &=\omega^{i r} h\left(\omega^{i}\right)^{s}=\omega^{a_{(i,1)}} \\
		g\left(\omega^{a_{(i,j)}}\right) &=\omega^{a_{(i,j)} r}h\left(\omega^{a_{(i,j)}}\right)^{s}=\omega^{a_{(i,j+1)}} \\
		g\left(\omega^{a_{(i,n-1)}}\right) &=\omega^{a_{(i,n-1)} r}h\left(\omega^{a_{(i,n-1)}}\right)^{s}=\omega^{i} . \\
	\end{aligned}\right.
\end{equation*}

Together with Theorem \ref{mulcoren}, $ f(x)=x^rh(x^s) $ is an $ n $-cycle permutation over $ \gf_q $ if and only if
\begin{equation*}
	\label{dairu}
	\left\{\begin{aligned}
		r^{n} \equiv & 1 \bmod s   \\
		\varphi(\omega^{a_{(i,j)}})=&\omega^{a_{(i,j)}(r^n-1)/s}   \prod\nolimits_{{k = 0}}^{n-1} {   h\left(     \omega^{a_{(i,j+k)}}    \right)^{r^{n-k-1}}      }     =1 ,\\
	\end{aligned}\right.
\end{equation*}
which is equivalent to
\begin{equation}
	\label{dairu2}
	\left\{\begin{aligned}
		r^{n} \equiv & 1 \bmod s   \\
		\varphi(\omega^{a_{(i,j)}})=&\beta^{a_{(i,j)}(r^n-1)}   \prod\nolimits_{{k= 0}}^{n-1} {       \beta^{r^{n-k-1}( \ell m_{a_{(i,j+k)}}+a_{(i,j+k+1)}-a_{(i,j+k)} r   ) }    }    =1 \\
	\end{aligned}\right.
\end{equation}
After simplifying Eq. (\ref{dairu2}), one can obtain that
\begin{equation}
	\label{dairu3}
	\left\{\begin{aligned}
		r^{n} \equiv & 1 \bmod s   \\
		\varphi(\omega^{a_{(i,j)}})=&\beta^{ \ell  \sum\nolimits_{k = 0}^{n-1} {r^{n-k-1}  m_{a_{(i,j+k)}}  }  } =1 \\
	\end{aligned}\right.
\end{equation}
It is easy to derive that Eq. (\ref{dairu3}) is equivalent to Eq. (\ref{solution}),
since $ \sum\nolimits_{k = 0}^{n-1} {r^{n-k-1}  m_{a_{(i,j+k)}}  }   \equiv  0 \bmod s  $ for $ 0 \le j \le n-1 $  is equivalent to
$$ \sum\nolimits_{k = 0}^{n-1} {r^{n-k-1}  m_{a_{(i,k)}}  }   \equiv  0 \bmod s  ,$$
 by applying $ r^{n} \equiv 1 \bmod s $.
Thus, $  f $ is an $ n $-cycle permutation if and only if Eq. (\ref{solution}) holds.
\end{proof}
In particular, we have the following consequence if $ g $ is the identity map.
\begin{Cor}
	\label{fenyuanc}
	Assume that $ a_{(i,j)} =a_{(i,j+1)}=i $ in Theorem \ref{piecewisegenerel}.
	Let $0 \leq m_{i} \leq s-1$ be an integer, and $ h(x)=\sum_{i=0}^{\ell-1} h_i x^i \in \gf_{q}[x] $ be defined as in Theorem \ref{piecewisegenerel}, where
	\begin{equation*}
		B=\left(\begin{array}{c}
			\beta^{\ell m_{0}} \\
			\vdots \\
			\beta^{\ell m_{i}+ i (1-r)} \\
			\vdots \\
			\beta^{\ell m_{\ell-1}+ (\ell-1)(1-r)}
		\end{array}\right) .
	\end{equation*}
	Then $  f (x) = x^r h(x^s) $ is an $ n $-cycle permutation over $ \gf_q $ if and only if
	\begin{equation}
		\label{solutioni}
		\left\{\begin{aligned}
			r^{n} \equiv & 1 \bmod s   \\
			 \sum\nolimits_{k = 0}^{n-k-1} {r^{n-k-1}  m_{i}  }   \equiv & 0 \bmod s . \\
		\end{aligned}\right.
	\end{equation}
\end{Cor}

\begin{proof}
		Since  $ a_{(i,j)} =a_{(i,j+1)}=i $ for $0 \leq i \leq \ell-1$, we have
	\begin{equation*}
	h\left(\omega^{i}\right) =\beta^{\ell m_{i}+i-i r} .
	\end{equation*}
	Hence,
\begin{equation*}
	g\left(\omega^{i}\right) =\omega^{i r} h\left(\omega^{i}\right)^{s}=\omega^{i} .
\end{equation*}
Plugging them into Theorem \ref{piecewisegenerel}, one can get $ f(x)=x^rh(x^s) $ is an $ n $-cycle permutation over $ \gf_q $ if and only if
\begin{equation*}
	\left\{\begin{aligned}
		r^{n} \equiv & 1 \bmod s   \\
		\varphi(\omega^{i})=&    \beta^{i(r^n-1)}   \prod\nolimits_{{k= 0}}^{n-1} {       \beta^{r^{n-k-1}( \ell m_{i}+i-ir   ) }    }       =   \beta^{ \ell  \sum\nolimits_{k = 0}^{n-1} {r^{n-k-1}  m_{i}  }  }  =1 ,\\
	\end{aligned}\right.
\end{equation*}
which is equivalent to Eq. (\ref{solutioni}). Therefore, we get the results.
\end{proof}

Although Theorems \ref{mulcoren} and \ref{piecewisegenerel} are from different perspectives, they are consistent in essence.
They will be both useful for judging and constructing $ n $-cycle permutations.

\section{$ n $-Cycle Permutations with High Index}
\label{high}

In this section, we will give explicit constructions of PPs of the form $ x^rh(x^s)  \in \gf_q[x]$ with high index $ \ell=\frac{q-1}{s} $.
It follows from  Theorem \ref{criterion}  and  (1) in Theorem \ref{mulcoren} that if $ r,s $ do not satisfy $r^n \equiv 1 \pmod{s}$ or $g(x)=x^rh(x)^s$ is not an $ n $-cycle permutation on $\mu_{\ell}$, clearly we have that $f(x) =x^rh(x^s)$ can not be an $ n $-cycle permutation on $ \gf_{q} $.
Therefore, when we let $r^n \equiv 1 \pmod{s}$ and control $ h(x) $ to construct an $ n $-cycle permutation $ f$, in addition to satisfy $h(x) \neq 0$ for any $x\in\mu_{\ell}$, we also must render $ g $ as an $ n $-cycle permutation.
We first consider constructions when $ g $ is a monomial, i.e., $g(x)=ax^v$, where $ a^{v^2+v+1}=1 $ and $ v^3-1 \equiv 0 \bmod s $.
Then, with $ h(x) $ keeping these above conditions, we adjust it to satisfy (2) in Theorem \ref{mulcoren}.

We mainly construct triple-cycle permutations in this section to demonstrate our approaches for constructing $ n $-cycle permutations.

\begin{Th}
	\label{single}
	Let $q$ be a prime power, $s \mid (q-1)$, $\gcd(r,s)=1$ and $r^{3} \equiv 1 \bmod s$.
	Assume that $h(x)\in \gf_q[x]$ such that $h(y)^s=ay^{v-r}$ holds for any $y \in \mu_\ell=\left\{    x\in{\gf}_{q}^*  \    |  \   x^\ell=1   \right\}$ , where $v^{3} \equiv 1 \bmod \ell$,  $a^{v^2+v+1}=1$ and $  \ell = {(q-1)/s} $.
	Then $f(x)= x^rh(x^s) $ is a triple-cycle permutation over $\mathbb{F}_{q}$  if and only if for any $y \in \mu_\ell$,
	$$ y^{(r^3-1)/s}h(y)^{r^2}h( ay^v )  ^rh\left(   a^{v+1}  y^{v^2}  \right)  =1 .$$
\end{Th}
\begin{proof}
	First, we have $g(x)=x^rh(x)^s=ax^v$ and  $h(y)^{sv}=a^vy^{v^2-rv}$.
Then, one can obtain
	\begin{equation*}
		\begin{aligned}
			& y^{(r^3-1)/s}h(y)^{r^2}h( ay^v )  ^rh\left(   a^{v+1}  y^{v^2}  \right)   \\
			= & y^{(r^3-1)/s}h(y)^{r^2}h( ay^v )  ^rh\left(   a  y^{rv}  h(y)^{sv}   \right)   \\
			= & y^{(r^3-1)/s}h(y)^{r^2}h\left( g\left( y\right) \right)  ^rh\left( g\left( y^{r}h(y)^s\right) \right)  \\
			= & \varphi(y).
		\end{aligned}
	\end{equation*}
	Thus, $f$ is a triple-cycle permutation if and only if $ y^{(r^3-1)/s}h(y)^{r^2}h( ay^v )  ^rh\left(   a^{v+1}  y^{v^2}  \right)   =1 $, according to Corollary \ref{mulcore}.
\end{proof}
A direct consequence of Theorem \ref{single} is the following.
\begin{Cor}
	\label{v=q}
	Let $q$ be a prime power, $s=q-1$,$\ell=(q^3-1)/(q-1)=q^2+q+1$.
	Assume that $h(x)\in \gf_{q^3}[x]$ such that $h(y)^{q-1}=1$ holds for any $y \in \mu_{q^2+q+1}=\left\{    x\in{\gf}_{q}^*  \    |  \   x^{q^2+q+1}=1   \right\}$.
	Then $f(x)= x^qh(x^{q-1}) $ is a triple-cycle permutation over $\mathbb{F}_{q^3}$  if and only if for any $y \in \mu_{q^2+q+1}$,
	$$  h(y)h( y^q )h\left(   y^{q^2}  \right)  =1 .$$
\end{Cor}

Due to the different structures of $ h(x) $, we will give three classes of triple-cycle permutations in the following.
For each class of them, we will also give explicit triple-cycle constructions.
\begin{Th}
	\label{jieguo4xiang3-3qp2}
	Let $ q $ be a prime power, $\phi(x) \in \mathbb{F}_{q^{2}}[x]$ and
	$$
	h(x)= \phi(x)+\phi(x)^{q} x^{1-v} ,
	$$
	where $v^{3} \equiv 1 \bmod (q+1)$.
	Then $f(x)=x h\left(x^{q-1}\right)$ is a triple-cycle permutation on $\mathbb{F}_{q^{2}}$ if and only if $ h(x)   h( x^v )   h\left(    x^{v^2}  \right)  =1 $ holds for any $x \in \mu_{q+1}$.
\end{Th}
\begin{proof}
	If there is an $ x_0 \in \mu_{q+1} $ such that $ h(x_0) = 0 $.
	Then $ h(x_0)^3  =0 \ne 1 $.
	Furthermore, $f(x_0)=x_0 h\left(x_0^{q-1}\right) =0$, thus $ f $ is not a triple-cycle permutation.
	
	If for any $ x\in \mu_{q+1} $, $ h(x) \ne 0 $.
	Then we have
	$$	h(x)^{q-1}=\frac{\phi(x)^q+\phi(x) x^{q-qv} }{\phi(x)+\phi(x)^{q} x^{1-v} } =x^{v-1}.	$$
	Thus, we obtain $g(x)=x h(x)^{q-1}=x^{v},$ which is a triple-cycle permutation on $\mu_{q+1}$.
	Then by plugging $a=1, r=1, s=q-1$ and $g(x)= x^{v}$ in the condition in Theorem \ref{single}, we obtain
	that $f(x)$ is a triple-cycle permutation on $\mathbb{F}_{q^{2}}$ if and only if 	$ h(x)   h( x^v )   h\left(    x^{v^2}  \right)  =1 $.
\end{proof}

\begin{Cor}
	\label{jieguo4xiang3-3qp2c}
	Let $ q $ be a power of $ 3 $ such that $ 1 + 3 q + 2q^2  \equiv    0  \pmod{q^3+1} $.
	Assume that
	$$ 	h(x)=  1 + {x^{1 + q}} + {x^{1 - {q^2}}} + {x^{ - {q^2} - q}}    . 	$$
	Then $f(x)=x h\left(x^{q^3-1}\right)$ is a triple-cycle permutation on $\mathbb{F}_{q^6}$.
\end{Cor}
\begin{proof}
	Let $\phi(x) = 1+x^{1+q}$ in Theorem \ref{jieguo4xiang3-3qp2}.
	Then, it suffices to prove $ h(x)   h( x^v )   h\left(    x^{v^2}  \right)  =1 $ holds for any $x \in \mu_{q+1}$, according to Theorem \ref{jieguo4xiang3-3qp2}.

	Since $ 1 + 3 q + 2q^2  \equiv    0  \pmod{q^3+1} $, one can obtain that $ 1 + 3 q + 2q^2  \equiv (1 + 3 q + 2q^2)q  \equiv (1 + 3 q + 2q^2)q^2  \equiv  0  \pmod{q^3+1} , $
	i.e., 	
	$$ 1 + 3 q + 2q^2  \equiv -2 q - 3 +q^2  \equiv q + 3 q^2 - 2  \equiv  0  \pmod{q^3+1}. $$
	Thus, for any $x \in \mu_{q+1}$,  we have
	\begin{equation}
		\label{huajie1}
		\left\{
		\begin{array}{ll}
			x^{-q - q^2}=x^{1 + 2 q + q^2} =x^{-2 + 2 q^2 }    \\
			x^{1 + q}= x^{-2 q - 2 q^2}=x^{-2 - q + q^2 }     \\
			x^{q + q^2}=x^{-1 - 2 q - q^2 }=x^{2 - 2 q^2 }  \\
			x^{2 + q - q^2}=  x^{-1 - q} =x^{2 q + 2 q^2}  .
		\end{array}
		\right.
	\end{equation}

	Therefore,  for any $x \in \mu_{q+1}$,	we have
	\begin{equation*}
		\begin{aligned}
			h(x)   h( x^{q^2} )   h\left(    x^{q^4}  \right) = &\left( x^{(q+1)q^3+1-{q^2}}+  x^{1-{q^2}}    +x^{q+1}  +1 \right)   \left(  x^{   (q+1)q^5+q^2-q^4}+  x^{{q^2}-q^4}    +x^{(q+1){q^2}}  +1 \right)  \\
			&  \cdot  \left( x^{   (q+1){q^7}+q^4-q^6}+  x^{q^4-q^6}    +x^{(q+1)q^4}  +1 \right)       \\
			=&1 + x^{1 + q} + x^{q^2  -1} + x^{ -q - q^2 } + x^{1 - q^2} + x^{1 +q} + x^{q^2 +q} + x^{1 + 2q + q^2 } + 1  + x^{	q^2 - 1}  \\
			&+ x^{ q + q^2 } + x^{-q - 1} + 1 + x^{ - q^2 -q} + x^{q + q^2}  + 1 + x^{	-q + q^2  - 2} + x^{ q^2-1} + x^{-q  -1} \\
			&+ x^{ - q^2-q} + 1 + x^{-1-q} + x^{- q^2 -2q -1} + x^{	-1-q} + x^{	-q -2-q} + x^{1 -q^2} + x^{1  -q-2q^2} \\
			&+ x^{1 -q^2}  + x^{q^2 -q^2} + x^{	1 + q } + x^{-1 -q} + x^{ q^2 -q^2} + x^{ -q-q^2}+ x^{	 -2q-2q^2}+ x^{ -q-q^2} \\
			&+ x^{-1 -2q-q^2} + x^{q -q^2 + 2} + x^{q + 1} + x^{2q  + 2} + x^{ -q^2 + 1} + 1+ x^{1 + q  } + x^{	 -q^2 + 1} \\
			&+ x^{-q -q^2 } + x^{ -2q^2 + 2} + x^{ -q^2 + 1} + x^{2 + q  -q^2 } + x^{ -q  -2q^2 + 1} + x^{ q+q^2} \\
			&+ x^{- 1 +q+2q^2} + x^{2q+2q^2} + x^{ - 1 +q^2} + x^{q+q^2} + 1 + x^{ 2q^2-2 } + x^{  2q^2-1 +q} + x^{ -1 +q^2} \\
			&+ x^{	-1+q^2} + x^{ -q -1 }  + x^{	-q -2+q^2} + x^{1 +q}  + x^{q+q^2} + x^{	1 +2q+q^2} + 1 \\
			=& 1, \\
		\end{aligned}
	\end{equation*}
	where the last step is simplified by Eq. (\ref{huajie1}).

	Thus, $f$ is a triple-cycle permutation on $\mathbb{F}_{q^6}$.

\end{proof}
\begin{example}
	Let $ q=3$, which satisfies $ 1 + 3\times 3 + 2\times 3^2 =28  \equiv    0  \pmod{3^3+1} $.
	Then $ h(x)=x^{20} + x^{16} + x^4 + 1 $.
	Thus $ f(x)=x^{521} + x^{313} + x^{105} + x $ is a triple-cycle permutation on $\gf_{3^{6}}$.
	These are verified by Magma.
\end{example}
The corollary below is obtained from corollary 	\ref{jieguo4xiang3-3qp2c} by applying the reductive way in Corollary \ref{diguiuse} .
\begin{Cor}
	\label{diguiuse1}
	Let $ n=3 $, $ q $ be a power of $ 3 $ such that $ 1 + 3 q + 2q^2  \equiv    0  \pmod{q^3+1} $, and $$ h(x)= 1 + x^{\frac{(1 + q)(q^3-1)+2q^6-2}{3}} + x^{\frac{(1 - {q^2})(q^3-1)+2q^6-2}{3}} + x^{\frac{( - {q^2} - q)(q^3-1)}{3}} .$$
	Then, $f(x)=x h\left(x^{\frac{q^{18}-1}{q^6-1}}\right)$ is a triple-cycle permutation on	$\mathbb{F}_{q^{18}}$.
\end{Cor}
\begin{proof}
	According to Corollary 	\ref{jieguo4xiang3-3qp2c},  		$g(x)=x  \left(   1 + {x^{(1 + q)(q^3-1)}} + {x^{(1 - {q^2})(q^3-1)}} + {x^{( - {q^2} - q)(q^3-1)}}  \right)  $ is a triple-cycle permutation on $\mathbb{F}_{q^6}$.
	From Corollary \ref{diguiuse}, $f$ is a triple-cycle permutation on	$\mathbb{F}_{q^{18}}$
\end{proof}

\begin{Th}
	\label{jieguo3xiang2-6qp2}
	Let $ q $ be a prime power, $\phi(x) \in \mathbb{F}_{q^{2}}[x]$ and $ h(x)= \phi(x)+\phi(x)^{q} +1. $
	Then $f(x)=x h\left(x^{q-1}\right)$ is a triple-cycle permutation on $\gf_{q^{2}}$ if and only if $ h(x)^3  =1 $ holds for any $x \in \mu_{q+1}$.
\end{Th}
\begin{proof}
	If there exists an $ x_0 \in \mu_{q+1} $ such that $ h(x_0) = 0 $.
	Then $ h(x_0)^3  =0 \ne 1 $.
	Furthermore, $f(x_0)=x_0 h\left(x_0^{q-1}\right) =0$, thus $ f $ is not a triple-cycle permutation.
	
	If for any $ x\in \mu_{q+1} $, $ h(x) \ne 0 $.
	Then we have
	$$
	h(x)^{q-1}=\frac{ \phi(x)+\phi(x)^{q} +1  }{ \phi(x)^{q}+\phi(x) +1 } =1
	$$
	Thus, we obtain $g(x)=x h(x)^{q-1}=x$, which is a triple-cycle permutation on $\mu_{q+1}$.
	Then by plugging $a=1, r=1,v=1, s=q-1$ and $g(x)= x$ in the condition in Theorem \ref{single}, we have  $f(x)$ is a triple-cycle permutation on $\mathbb{F}_{q^{2}}$ if and only if
	$ h(x)^3  =1 $ for any $ x \in \mu_{q+1} $.
\end{proof}

\begin{Cor}
	\label{jieguo3xiang2-6qp2c}
	Let $ q $ be an even prime power, and $ a $ be an integer such that $ 5a  \equiv 0 \pmod{q+1}  $.
	Then $f(x)=x h\left(x^{q-1}\right)$ is a triple-cycle permutation on $\gf_{q^2}$, where $ 	h(x)=  x^a+x^{aq}  +1    $.
\end{Cor}
\begin{proof}
	For any $ x \in \mu_{q+1} $, we have $x^{aq}=x^{a + 2aq}, x^{a}=x^{2a + aq}, x^{2a}=x^{3aq}$ , and $x^{3a}=x^{ 2aq}$. Since $ 5a  \equiv 0 \pmod{q+1}  $,
then one can obtain that
	\begin{equation*}
		\begin{aligned}
			h(x)^3 &=    1 + {x^a} + {x^{2a}} + {x^{3a}} + {x^{aq}} + {x^{2aq}} + {x^{3aq}} + {x^{2a + aq}} + {x^{a + 2aq}} \\
			&=    1 + {x^a} + {x^{2a}} + {x^{3a}} + {x^{aq}} + {x^{3a}} + {x^{2a}} + {x^{a}} + {x^{aq}} \\
			&= 1. \\
		\end{aligned}
	\end{equation*}
	Thus $ f $ is a triple-cycle permutation.
\end{proof}
\begin{example}
	Let $ q=2^6, a_1=26, a_2=13$, $h_1(x)=x^{39} + x^{26 }+ 1,$ and $h_2(x)=x^{52} + x^{13} + 1 $.
	Then we have $ 5\times 26  \equiv 0 \pmod{2^6+1} $, $ 26 \times 2^6  \equiv 39 \pmod{2^6+1} $, $ 5\times 13  \equiv 0 \pmod{2^6+1} $ and $ 13 \times 2^6  \equiv 52 \pmod{2^6+1} $.
	Thus $ f_1(x)=x^{2458} + x^{1639} + x $ and $ f_2(x)=x^{3277} + x^{820} + x  $ are triple-cycle permutations on $\gf_{q^{12}}$.
	These are verified by Magma.
\end{example}

The corollary below is  obtained from corollary 	\ref{jieguo3xiang2-6qp2c} by applying the reductive way in Corollary \ref{diguiuse}.
\begin{Cor}
	\label{diguiuse2}
	Let $ m=2 $ in Corollary \ref{diguiuse}, $ q $ be a power of $ 2 $, and $ a $ be an integer such that $ 5a  \equiv 0 \pmod{q+1}  $.
	Assume that $$ h(x)= x^{a(q-1)}+x^{aq(q-1)}  +1  . $$
	Then $f(x)=x h\left(x^{q^2+1}\right)$ is a triple-cycle permutation on	$\mathbb{F}_{q^{4}}$.
\end{Cor}
\begin{proof}
		According to Corollary 	\ref{jieguo3xiang2-6qp2c}, 	$g(x)=x  \left(    x^{a(q-1)}+x^{aq(q-1)}  +1   \right)  $ is a triple-cycle permutation on $\mathbb{F}_{q^2}$.
			According to Corollary \ref{diguiuse}, $f$ is a triple-cycle permutation on	$\mathbb{F}_{q^{4}}$.
\end{proof}

\begin{Th}
	\label{jieguo3xiang2-6qp2-2}
	Let $ q $ be a prime power, $\phi(x) \in \mathbb{F}_{q^{2}}[x]$ such that $  \phi(x)^{q-1}=x^{v-1} $ for any $x \in \mu_{q+1}$.
	Let	$$  h(x)=\phi(x)+1 + x^{1-v} ,	$$
	where $v^{3} \equiv 1 \bmod (q+1)$.
	Then $f(x)=x h\left(x^{q-1}\right)$ is a triple-cycle permutation on $\mathbb{F}_{q^{2}}$ if and only if $ h(x)   h( x^v )   h\left(    x^{v^2}  \right)  =1 $ holds for any $x \in \mu_{q+1}$.
\end{Th}
\begin{proof}
	Assume that for any $ x\in \mu_{q+1} $, $ h(x) \ne 0 $, since it is a necessary condition for both $f(x)=x h\left(x^{q-1}\right)$ being a triple-cycle permutation and $ h(x)   h( x^v )   h\left(    x^{v^2}  \right)  =1 $.
	Clearly $  q-qv \equiv   v-1   \pmod{q+1}$.
	Then we have
	$$	h(x)^{q-1}=\frac{      \phi(x)^q+       x^{q-qv}  +1    }{   \phi(x)  + 1 + x^{1-v}    } =x^{v-1}.	$$
	Thus, one can obtain that $g(x)=x h(x)^{q-1}=x^{v},$  which is a triple-cycle permutation on $\mu_{q+1}$.
	Then by plugging $a=1, r=1, s=q-1$ and $g(x)= x^{v}$ in the condition in Theorem \ref{single}, we obtain
	that $f(x)$ is a triple-cycle permutation on $\mathbb{F}_{q^{2}}$ if and only if 	$ h(x)   h( x^v )   h\left(    x^{v^2}  \right)  =1 $.
\end{proof}
\begin{Cor}
	Let $ q $ be an even prime power, and $ a $ be an integer such that
\begin{subequations}
\begin{align}
			 a (1+ v + v^2) \equiv&  0   \pmod{q+1}\label{pinhen}, \\
		a + v - v^2 + av^2-av  \equiv & 0  \pmod{q+1}\label{yong1}  \text{ and}    \\
		 av + v^2 +v -2 \equiv  & 0  \pmod{q+1}.\label{yong2}
	\end{align}
are all established.
\end{subequations}
	Then $f(x)=x h\left(x^{q-1}\right)$ is a triple-cycle permutation on $\gf_{q^2}$, where $ h(x)=x^a +1 + x^{1-v}    $.
\end{Cor}
\begin{proof}
	By simplifying Eq. (\ref{yong1}) $ + $ Eq. (\ref{yong2}) $ - $ Eq. (\ref{pinhen}), we have
	\begin{equation}
		\label{yong3}
		-2 -av +2v  \equiv  0  \pmod{q+1}.
	\end{equation}
  Respectively multiplying Eq. (\ref{yong1}), Eq. (\ref{yong2}) and Eq. (\ref{yong3}) with $ v $ and $ v^2$, and then simplifying them by (\ref{pinhen}), we have
	$$  av + v^2 - 1 + a  -av^2 \equiv  av^2 + 1 - v + av  -a \equiv   0  \pmod{q+1}, $$
	$$ av^2 + 1 +v^2 -2v \equiv  a + v +1  - 2v^2 \equiv   0  \pmod{q+1}, $$
	$$   -2v  - av^2+ 2v^2 \equiv  -2v^2-a +2 \equiv   0  \pmod{q+1}. $$
	Therefore, one can obtain that
	\begin{equation*}
		\begin{aligned}
			h(x)   h(x^v)   h(x^{v^2}) = &\left( x^a + 1 + x^{1 - v} \right)   \left(  x^{av} + 1 + x^{v - v^2} \right)  \left( x^{av^2} + 1 + x^{v^2 - 1} \right)       \\
			=&x^a + x^{1 - v} + x^{-1 + v} + x^{a v} + x^{a v^2} + x^{-1 + a + v} + x^{a + a v} + x^{1 - v + a v} + x^{1 - v^2} + x^{v - v^2} \\
			&+ x^{a + v - v^2} + x^{-1 + v^2} + x^{-1 + a + v^2} + x^{-v + v^2}  + x^{-1 + a v + v^2} + x^{-1 + a + a v + v^2} + x^{-v + a v + v^2} \\
                      			&+ x^{a + a v^2} + x^{1 - v + a v^2} + x^{a v + a v^2} + x^{a + a v + a v^2} + x^{1 - v + a v + a v^2} + x^{1 - v^2 + a v^2} + x^{v - v^2 + a v^2} \\
			&+ x^{a + v - v^2 + a v^2}  \\
			=& 1. \\
		\end{aligned}
	\end{equation*}
	Thus $ f $ is a triple-cycle permutation.
\end{proof}
\begin{example}
	Let $ q=2^6, a_1=35,v_1=61, a_2=25, v_2=16$, $h_1(x)=x^{35} + x^5 + 1,$ and $ h_2(x)= x^{25} + x^{50}  + 1$.
	It is easy to derive that $ 35 (1+ 61 +61^2) \equiv 35 + 61 - 61^2 + 35 \times61^2-35\times 61 \equiv  35\times61 + 61^2 +61 -2 \equiv 2 +35\times61 -2\times 61 \equiv  0  \pmod{q+1} $   and
	$ 25 (1+ 16 + 16^2) \equiv 25 + 16 - 16^2 + 25\times 16^2-25\times 16 \equiv  25\times 16 + 16^2 +16 -2 \equiv 2 +25\times 16 -2\times 16 \equiv  0  \pmod{q+1} $.
	Thus $ f_1(x)=x^{2206} + x^{316} + x $ and $ f_2(x)=x^{1576} + x^{3151}  + x $ are triple-cycle permutations on $\gf_{q^{12}}$.
	These are verified by Magma.
\end{example}

\section{$ n $-Cycle Permutations with Low Index }
\label{low}

In this section, we use piecewise method to construct explicit $ n $-cycle permutations for low index $ \ell $.
All the constructions in this section are consistent with Theorem \ref{piecewisegenerel}.

For explicit constructions with low index, we first discuss a special case that $\mu_{\ell}$ having at most two elements, where $s=\frac{q-1}{2}$ and $ q $ is an odd prime power.
\begin{Th}
	\label{miu2}
	Let $q$ be an odd prime power, and
	 $$ f(x)=\frac{a-b}{2} x^{\frac{q-1}{2}+r}+\frac{a+b}{2} x^{r} \in \gf_q[x],$$
	where $ a,b \in \gf_q  $.
	Assume that $ s=\frac{q-1}{2}  $.
	Then $f$ is an $ n $-cycle permutation over $\mathbb{F}_{q}$ if and only if
	\begin{enumerate}[(1)]
		\item $r^{n} \equiv 1 \bmod s$,
		\item  $ a^{\sum\nolimits_{{i = 0}}^{n-1} {   r^i      }}={-1}^{(r^n-1)/s}   b^{\sum\nolimits_{{i = 0}}^{n-1} {   r^i      }}=1 $ or \\
	 $ n $ is even and $ a^{\sum\nolimits_{{k=1}}^{n/2} {   r^{2k-1}      } }   b^{\sum\nolimits_{{k=1}}^{n/2} {   r^{ 2k-2 }      }}   ={(-1)}^{(r^n-1)/s}    a^{\sum\nolimits_{{k=1}}^{n/2} {   r^{ 2k-2 }      }}  b^{\sum\nolimits_{{k=1}}^{n/2} {   r^{2k-1}      } }    =1   .$
	\end{enumerate}
\end{Th}
\begin{proof}
As we know,	$h(x)=\frac{a-b}{2} x^{\frac{q-1}{2}}+\frac{a+b}{2}$,  and $ a=h(1), b=h(-1)  $.
	Then, $f(x)$ can be rewritten as
	\begin{equation}
		\label{twoelements}
		f(x)=x^{r} h\left(x^{s}\right)=\left\{
		\begin{array}{ll}
			0, & x=0 \\
			a x^{r}, & x \in S_{1} \\
			b x^{r}, & x \in S_{-1} \\
		\end{array}\right.
	\end{equation}
	where $ S_{\pm 1}=\left\{y \in \mathbb{F}_{q}^* | y^{\frac{q-1}{3}}= \pm 1 \right\}$.
	
	If $ n  $ is odd.
	When $ f $ is an $ n $-cycle permutation on $ \gf_q $, one can obtain that $ g(x)=x $, since $ g $ must be an $ n $-cycle permutation on $ \mu_2 $.
	By plugging them into Theorem \ref{mulcoren}, we have
	$$\varphi(1) = \prod\nolimits_{{i = 0}}^{n-1} {   h\left(1 \right)^{r^{n-i-1}}      }   = a^{\sum\nolimits_{{i = 0}}^{n-1} {   r^i      }}, $$
	$$	\varphi(-1)={(-1)}^{(r^n-1)/s}       \prod\nolimits_{{i = 0}}^{n-1} {   h\left(   -1 \right)^{r^{n-i-1}}      }    	= {(-1)}^{(r^n-1)/s}   b^{\sum\nolimits_{{i = 0}}^{n-1} {   r^i      }}. $$
	
	If $ n $ is even.
	It suffices to consider that $ g $ is an involution on $ \mu_2 $, i.e., $ g(1)=-1, g(-1)=1 $.
	By plugging them into Theorem \ref{mulcoren}, one can obtain that
	$$\varphi(1) =   \prod\nolimits_{{k=1}}^{n/2} {   h\left(1 \right)^{r^{n-2k+1}}      }     h\left(  -1 \right)^{r^{n-2k}}        = a^{\sum\nolimits_{{k=1}}^{n/2} {   r^{2k-1}      } }   b^{\sum\nolimits_{{k=1}}^{n/2} {   r^{ 2k-2 }      }}  , $$   
	$$	\varphi(-1)={(-1)}^{(r^n-1)/s}   \prod\nolimits_{{k=1}}^{n/2} {   h\left(-1 \right)^{r^{n-2k+1}}      }    h\left(  1 \right)^{r^{n-2k}}     	= {(-1)}^{(r^n-1)/s}    a^{\sum\nolimits_{{k=1}}^{n/2} {   r^{ 2k-2 }      }}  b^{\sum\nolimits_{{k=1}}^{n/2} {   r^{2k-1}      } }     . $$

	According to Theorem \ref{mulcoren},  the result is established.
\end{proof}

A consequence of Theorem \ref{miu2} is the following.
\begin{Cor}
	\label{3miu2}
	Let $q$ be an odd prime power,
	and $$ f(x)=\frac{a-b}{2} x^{\frac{q-1}{2}+r}+\frac{a+b}{2} x^{r} \in \gf_q[x],$$
	where $ a,b \in \gf_q  $.
	Assume that $  s=\frac{q-1}{2}  $.
	Then $f$ is a triple-cycle permutation over $\mathbb{F}_{q}$ if and only if
	\begin{enumerate}[(1)]
		\item $r^{3} \equiv 1 \bmod s$,
		\item $ a^{r^2+r+1}  = {(-1)}^{(r^3-1)/s}   b^{r^2+r+1}    =1 $.
	\end{enumerate}
\end{Cor}
%
Very recently, Wu et al. \cite[Theorem 2]{wuCharacterizationsConstructionsTriplecycle2020a} also considered such binomial triple-cycle permutations.
By applying Corollary \ref{mulcore}, they discussed the results in three cases.
However, we noticed that an $ n $-cycle permutation on $ \mu_2 $ with odd $ n $ can only be the identity map.
Hence, we simplify all cases in \cite[Theorem 2]{wuCharacterizationsConstructionsTriplecycle2020a} into only one situation and obtain a concise condition.
We also have an explicit construction for Corollary \ref{3miu2}, as an example.
\begin{example}
	Let $r=1$, $ a \in \gf_q \setminus \{1\}$ and $ b= a^2$ such that $ a^{3}  =1 $.
	Then all the conditions in Corollary \ref{3miu2} are satisfied.
	Thus $ f(x)=\frac{a-a^2}{2} x^{\frac{q-1}{2}+1}+\frac{a+a^2}{2} x $ is a triple-cycle permutation over $ \gf_q $.
\end{example}
Another consequence of Theorem \ref{miu2} is the following.
\begin{Cor}
	\label{4miu2}
	Let $q$ be an odd prime power, and
	 $$ f(x)=\frac{a-b}{2} x^{\frac{q-1}{2}+r}+\frac{a+b}{2} x^{r} \in \gf_q[x],$$
	where $ a,b \in \gf_q  $.
	Assume that $ s=\frac{q-1}{2}  $.
	Then $f$ is a quadruple-cycle permutation over $\mathbb{F}_{q}$ if and only if
	\begin{enumerate}[(1)]
		\item $r^{4} \equiv 1 \bmod s$,
		\item  $ a^{   r^3+r^2+r+1   }={(-1)}^{(r^4-1)/s}   b^{   r^3+r^2+r+1   }=1 $ or \\
		$ a^{r^3+r}    b^{   r^2+1   }   ={(-1)}^{(r^4-1)/s}       a^{   r^2+1   } b^{r^3+r}  =1   .$
	\end{enumerate}
\end{Cor}
\begin{example}
	Let $r=1$, $ ab \in \gf_q \setminus \{1\}$ such that $ a^2b^2  =1 $.
	Then all the conditions in Corollary \ref{4miu2} are satisfied.
	Thus $ f(x)=\frac{a-a^2}{2} x^{\frac{q-1}{2}+1}+\frac{a+a^2}{2} x $ is a quadruple-cycle permutation over $ \gf_q $.
\end{example}

Now we discuss the case that the amount of $\mu_{\ell}$ is exactly  three.

\begin{Th}
	\label{nmiu3}
	Let $q$ be a prime power such that $s=\frac{q-1}{3}$ is an integer, and $ \alpha $ be a primitive element of $ \gf_q $.
	Let $$ h(x)=\frac{b-c-a \omega^{2}+b \omega^{2}}{(-1+\omega)^{2} \omega} x^{2}+     \frac{c+a \omega-b(1+\omega)}{(-1+\omega)^{2} \omega(1+\omega)}x+     \frac{c+a \omega^{3}-b \omega(1+\omega)}{(-1+\omega)^{2}(1+\omega)}    ,$$
	where $ a,b,c \in \gf_q  $, $\omega=\alpha^{\frac{q-1}{3}}$.
	Assume that $ g(x)=x^rh(x)^s=x $ is on $ \mu_{3} $. If
	\begin{enumerate}[(1)]
		\item $r^{n} \equiv 1 \bmod s$ and
		\item $a^{\sum\nolimits_{{i = 0}}^{n-1} {   r^i      }}     =   \omega^{(r^n-1)/s} b^{\sum\nolimits_{{i = 0}}^{n-1} {   r^i      }}  = \omega^{2(r^n-1)/s}  c^{\sum\nolimits_{{i = 0}}^{n-1} {   r^i      }}   $,
	\end{enumerate}
	then $f(x)=x^rh(x^s)$ is an $ n $-cycle permutation over $\mathbb{F}_{q}$.
\end{Th}
\begin{proof}
It is easy to derive that	$ a=h(1), b=h(\omega), c=h(\omega^2)  $.
	Then, $f(x)$ can be rewritten as
	\begin{equation}
		\label{}
		f(x)=x^{r} h\left(x^{s}\right)=\left\{
		\begin{array}{ll}
			0, & x=0 \\
			h(1) x^{r}, & x \in S_{1} \\
			h(\omega) x^{r}, & x \in S_{\omega} \\
			h(\omega^2) x^{r}, & x \in S_{\omega^2},
		\end{array}\right.
	\end{equation}
	where $ S_{i}=\left\{y \in \mathbb{F}_{q}^* | y^{\frac{q-1}{3}}=i \right\} $ for $ i=1,\omega,\omega^2 $.
	
	Plugging $ g(x)=x $ into Theorem \ref{mulcoren}, one can obtain that
	$$a^{\sum\nolimits_{{i = 0}}^{n-1} {   r^i      }}     =   \omega^{(r^n-1)/s} b^{\sum\nolimits_{{i = 0}}^{n-1} {   r^i      }}  = \omega^{2(r^n-1)/s}  c^{\sum\nolimits_{{i = 0}}^{n-1} {   r^i      }} =1,  $$
	which is equivalent to 
	$$\varphi(1) = \prod\nolimits_{{i = 0}}^{n-1} {   h\left(1 \right)^{r^{n-i-1}}      }  =1 , $$
	$$\varphi(\omega)=\omega^{(r^n-1)/s}         \prod\nolimits_{{i = 0}}^{n-1} {   h\left(\omega \right)^{r^{n-i-1}}      }   =1 ,$$
	$$\varphi(\omega^2)=\omega^{2(r^n-1)/s} \prod\nolimits_{{i = 0}}^{n-1} {   h\left(\omega^2 \right)^{r^{n-i-1}}      }   =1 .$$
	Thus the result is obtained.
\end{proof}
We also have an explicit construction of quadruple-cycle permutation as a consequence of Theorem \ref{nmiu3}.
\begin{example}
	Let $ n=4 $ and $ r=1 $ in Theorem \ref{nmiu3}.
	Let $q$ be a prime power such that $s=\frac{q-1}{3}$ is an integer, and $ \beta $ be a primitive root of $ 1 $.
	Assume that $ a=\beta,b=\beta^2,c=\beta^3 $.
	Then, all the conditions in Theorem \ref{nmiu3} are satisfied, and
	$$ f(x)=    \frac{b-c-a \omega^{2}+b \omega^{2}}{(-1+\omega)^{2} \omega} x^{\frac{2q-2}{3}+1}+     \frac{c+a \omega-b(1+\omega)}{(-1+\omega)^{2} \omega(1+\omega)}x^{\frac{q-1}{3}+1}+     \frac{c+a \omega^{3}-b \omega(1+\omega)}{(-1+\omega)^{2}(1+\omega)} x       $$
	is a quadruple-cycle permutation over $ \gf_q $.
\end{example}

For such trinomial in Theorem \ref{nmiu3}, we also give an explicit necessary and sufficient condition for it to be a triple-cycle permutation.
\begin{Th}
	\label{miu3}
	Let $q$ be a prime power such that $s=\frac{q-1}{3}$ is an integer, and $ \alpha $ be a primitive element of $ \gf_q $.
	Let $$ h(x)=\frac{b-c-a \omega^{2}+b \omega^{2}}{(-1+\omega)^{2} \omega} x^{2}+     \frac{c+a \omega-b(1+\omega)}{(-1+\omega)^{2} \omega(1+\omega)}x+     \frac{c+a \omega^{3}-b \omega(1+\omega)}{(-1+\omega)^{2}(1+\omega)}    ,$$
	where $ a,b,c \in \gf_q  $, $\omega=\alpha^{\frac{q-1}{3}}$.
	Assume that $ g(x)=x^rh(x)^s $ is on $ \mu_{3} $.
	Then $f(x)=x^rh(x^s)$ is a triple-cycle permutation over $\mathbb{F}_{q}$ if and only if
	\begin{enumerate}[(1)]
		\item $r^{3} \equiv 1 \bmod s$ and
		\item $ a^{r^2}b^r   c  =\omega^{(r^3-1)/s}b^{r^2}c^r  a=\omega^{2(r^3-1)/s}c^{r^2} a^r b =1$
		or \\ $a^{r^2+r+1}=\omega^{(r^3-1)/s}b^{r^2+r+1}=\omega^{2(r^3-1)/s}c^{r^2+r+1}=1$.
	\end{enumerate}
\end{Th}
\begin{proof}
	It is easy to get that	$ a=h(1), b=h(\omega), c=h(\omega^2)  $.
	Then, $f(x)$ can be rewritten as
\begin{equation}
	\label{threeelements}
	f(x)=x^{r} h\left(x^{s}\right)=\left\{
	\begin{array}{ll}
		0, & x=0 \\
		h(1) x^{r}, & x \in S_{1} \\
		h(\omega) x^{r}, & x \in S_{\omega} \\
		h(\omega^2) x^{r}, & x \in S_{\omega^2},
	\end{array}\right.
\end{equation}
where $ S_{i}=\left\{y \in \mathbb{F}_{q}^* | y^{\frac{q-1}{3}}=i \right\} $ for $ i=1,\omega,\omega^2 $.

	If $ g(x)=x $, plugging it into Corollary \ref{mulcore}, one can obtain that
$$\varphi(1)=h(1)^{r^2}h\left(  1 \right) ^rh\left( 1\right)  =1 $$
$$\varphi(\omega)=\omega^{(r^3-1)/s}h(\omega)^{r^2}h\left( \omega  \right)  ^rh\left( \omega  \right)  =1 $$
$$\varphi(\omega^2)=\omega^{2(r^3-1)/s}h(\omega^2)^{r^2}h\left( \omega^2 \right)  ^rh\left(  \omega^2 \right)  =1 ,$$
which is equivalent to  $a^{r^2+r+1}=\omega^{(r^3-1)/s}b^{r^2+r+1}=\omega^{2(r^3-1)/s}c^{r^2+r+1}=1$.

	If $ g $ is not the identity map, as it must be a triple-cycle permutation on $ \mu_{3} $,  it suffices to assume that $ g(1)=\omega, g(\omega)=\omega^2, g(\omega^2)=1 $.
	By plugging them into Corollary \ref{mulcore}, we have
	$$\varphi(1)=h(1)^{r^2}h\left( \omega \right)  ^r   h\left( \omega^2   \right) =a^{r^2}b^r   c  $$
	$$\varphi(\omega)=\omega^{(r^3-1)/s}h(\omega)^{r^2}h\left( \omega^2 \right) ^rh\left( 1 \right)  = \omega^{(r^3-1)/s}b^{r^2}c^r  a$$
	$$\varphi(\omega^2)=\omega^{2(r^3-1)/s}h(\omega^2)^{r^2}h\left( 1\right)  ^rh\left(    \omega \right)  =\omega^{2(r^3-1)/s}c^{r^2} a^r b  $$
	According to Corollary \ref{mulcore},  $f$ is a triple-cycle permutation over $\mathbb{F}_{q}$ if and only if
		\begin{enumerate}[(1)]
		\item $r^{3} \equiv 1 \bmod s$ and
		\item  $ a^{r^2}b^r   c  =\omega^{(r^3-1)/s}b^{r^2}c^r  a=\omega^{2(r^3-1)/s}c^{r^2} a^r b =1$
		or \\ $a^{r^2+r+1}=\omega^{(r^3-1)/s}b^{r^2+r+1}=\omega^{2(r^3-1)/s}c^{r^2+r+1}=1$.
	\end{enumerate}
\end{proof}

We also have an explicit construction in this case.
\begin{example}
	Let $r=1$, $ a ,b,c $ be any elements of $\gf_q^* $ such that $ abc =1 $.
	Then, all the conditions in Theorem \ref{miu3} are satisfied.
	Thus $$ f(x)=    \frac{b-c-a \omega^{2}+b \omega^{2}}{(-1+\omega)^{2} \omega} x^{\frac{2q-2}{3}+1}+     \frac{c+a \omega-b(1+\omega)}{(-1+\omega)^{2} \omega(1+\omega)}x^{\frac{q-1}{3}+1}+     \frac{c+a \omega^{3}-b \omega(1+\omega)}{(-1+\omega)^{2}(1+\omega)} x       $$ is a triple-cycle permutation over $ \gf_q $.
\end{example}

\section{CONCLUSIONS}
\label{cons}

In this paper, we studied the properties and constructions of $n$-cycle permutations.
We proposed unified methods for determining and constructing $ n $-cycle permutations of the form $x^rh(x^s)$, including a criterion, a reductive constructing approach and the cyclotomic approach.
Five classes of $ n $-cycle permutations are constructed explicitly, and some of them are triple-cycle permutations.
It is interesting to use the methods in this paper to obtain more  $ n $-cycle permutations of other forms in the future.

\bibliographystyle{plainnat}  

\end{document}